\documentclass[twoside,12pt]{article}
\usepackage{CJK}
\usepackage{indentfirst, color}
\usepackage{bm}
\usepackage{graphicx}
\usepackage{epsfig}
\usepackage{amsmath}
\usepackage{amsfonts}
\usepackage{amssymb}
\usepackage{amsthm}
\usepackage[superscript]{cite}
\usepackage{latexsym}

\newtheorem{thm}{Theorem}
\newtheorem{lemma}{Lemma}

\usepackage{epsfig}
\usepackage{epstopdf}
\usepackage{pgf,fancyhdr}
\usepackage{float}
\usepackage{amsmath}

\begin{document}

\begin{CJK*}{GBK}{song}

\begin{center}
\LARGE\bf Tighter monogamy relations of entanglement measures based on fidelity
\end{center}

\begin{center}
\rm  Mei-Ming Zhang,$^1$ \  Naihuan Jing,$^{1, 2,*}$
\end{center}

\begin{center}
\begin{footnotesize} \sl
$^1$ Department of Mathematics, Shanghai University, Shanghai 200444, China  %

$^2$ Department of Mathematics, North Carolina State University, Raleigh, NC 27695, USA

$^*$ Corresponding author: jing@ncsu.edu

\end{footnotesize}
\end{center}

\begin{center}
\begin{minipage}{15.5cm}
\parindent 20pt\footnotesize
We study the Bures measure of entanglement and the geometric measure of entanglement as special cases of entanglement measures based on fidelity, and find their tighter monogamy inequalities over tri-qubit systems as well as multi-qubit systems. 
Furthermore, we derive the monogamy inequality of concurrence for qudit quantum systems by projecting higher-dimensional states to qubit substates.
\end{minipage}
\end{center}

\begin{center}
\begin{minipage}{15.5cm}
\begin{minipage}[t]{2.3cm}{\bf Keywords:}\end{minipage}
\begin{minipage}[t]{13.1cm}
Monogamy, Bures measure of entanglement, Geometric measure of entanglement, Concurrence
\end{minipage}\par\vglue8pt

\end{minipage}
\end{center}

\section{Introduction}
Quantum entanglement is an indispensable resource for quantum information processing [1] which distinguishes quantum mechanics from the classical one.
In contrast to classical correlations, quantum entanglement has an interesting feature in that entanglement cannot be freely shared among systems. For instance, if two parties are maximally entangled in a multi-partite systems, then none of them could share entanglement with any part of the remaining system.
We call this phenomenon quantum entanglement monogamy [2].
The monogamy relation of entanglement serves to characterize different kinds of entanglement distribution.

Entanglement monogamy was first characterized as an inequality in three-qubit systems by Coffman-Kundu-Wootters (CKW) [3], i.e, $E(\rho_{A|BC})\geq E(\rho_{AB})+E(\rho_{AC})$ where $E(\rho_{A|BC})=C^2(\rho_{A|BC})$ represents the squashed concurrence of $\rho_{A|BC}$ under bipartition ${A}$ and $BC$, $\rho_{AB}$ and $\rho_{AC}$ are the reduced density matrices of the tri-qubit state $\rho_{ABC}$ respectively.
The concurrence of a two-qubit mixed state $\rho$ is given by the analytic formula $C(\rho)=\max\{\lambda_1-\lambda_2-\lambda_3-\lambda_4, 0\}$, where $\lambda_{1},\lambda_{2},\lambda_{3},\lambda_{4}$, are the square roots of nonnegative eigenvalues of the matrix $\rho(\sigma_y\otimes\sigma_y)\rho^*(\sigma_y\otimes\sigma_y)$ arranged in nonincreasing order, $\sigma_y$ is the Pauli matrix, and $\rho^*$ denotes the complex conjugate of $\rho$ [3].
Later, Osborne and Verstraete proved that the CKW inequality also holds in an n-qubit system [4].
Other types of monogamy relations for entanglement were also proposed, notably the entanglement negativity [5-8], entanglement of formation [9,10], Tsallis q-entropy [11-13], R\'{e}nyi-$\alpha$ entanglement [14-16] and unified-$(q,s)$ entanglement [17].
In recent years, monogamy inequalities of one class of entanglement measures based on fidelity, such as the Bures measure of entanglement [18, 19] and the geometric measure of entanglement [20] were discussed in [21, 22].
All these monogamy relations were basically presented for qubit quantum states, however in higher-dimensional systems, some quantum states were found violating the CKW inequality [23, 24].
The monogamy relation of an entanglement measure for qudits only was also given with a strong conjecture that seems to have no obvious counterexamples [6].
In this paper, we generalize the monogamy inequality in qudit quantum systems by projecting higher-dimensional states to qubit substates.

This article is organized as follows. In Section 2, we derive the tightened monogamy inequalities in an arbitrary tripartite mixed state based on the Bures measure of entanglement and the geometric measure of entanglement. Then the monogamy relation is generalized to multipartite quantum systems.
By detailed examples, our results are seen to be superior to the previously published results.
In Section 3, we derive the monogamy inequalities of concurrence in an arbitrary dimensional tripartite systems by projecting high-dimensional states to $2\otimes 2\otimes 2$ sub-states and we generalize the results for the multipartite quantum systems. Comments and conclusions are given in Section 4.

\section{Tighter monogamy relations of entanglement measures based on fidelity}
Recall that the fidelity of separability is defined by [25]:
\begin{eqnarray}
F_S(\rho_{A_1A_2})=\max_{\sigma_{A_1A_2}\in S}F(\rho_{A_1A_2},\sigma_{A_1A_2})
\end{eqnarray}
where $S$ is the set of separable states, the maximum is taken over all separable states $\sigma_{A_1A_2}$ in S and $F(\rho_{A_1A_2}, \sigma_{A_1A_2})=[tr((\sqrt{\rho_{A_1A_2}}\sigma_{A_1A_2}\sqrt{\rho_{A_1A_2}})^{\frac12})]^2$.
Now we consider the entanglement measures based on fidelity for the Bures measure of entanglement and the geometric measure of entanglement, which are defined respectively by [18, 20]:
\begin{eqnarray}
M_B(\rho_{A_1A_2})&=&\min_{\sigma_{A_1A_2}\in S}(2-2\sqrt{F(\rho_{A_1A_2},\sigma_{A_1A_2})})=2-2\sqrt{F_S(\rho_{A_1A_2})},\\
M_G(\rho_{A_1A_2})&=&\min_{\sigma_{A_1A_2}\in S}(1-F(\rho_{A_1A_2},\sigma_{A_1A_2}))=1-F_S(\rho_{A_1A_2}).
\end{eqnarray}
For an arbitrary two-qubit mixed state, the analytical expressions for the Bures measure of entanglement and the geometric
measure of entanglement in terms of the concurrence are given as follows [25]:
\begin{eqnarray}
M_B(\rho_{A_1A_2})&=&f_B(C(\rho_{A_1A_2}))=2-2\sqrt{\frac{1+\sqrt{1-C^2(\rho_{A_1A_2})}}{2}},\\
M_G(\rho_{A_1A_2})&=&f_G(C(\rho_{A_1A_2}))=\frac{1-\sqrt{1-C^2(\rho_{A_1A_2})}}{2},
\end{eqnarray}
where $f_B(x)=2-2\sqrt{\frac{1+\sqrt{1-x^2}}{2}}$ and $f_G(x)=\frac{1-\sqrt{1-x^2}}{2}$ are monotonically increasing functions in $0\leq x\leq1$. For $2\otimes d$ $(d\geq2)$ mixed states $\rho_{A_1A_2}$, one has the relations $M_B(\rho_{A_1A_2})\geq f_B(C(\rho_{A_1A_2}))$ and $M_G(\rho_{A_1A_2})\geq f_G(C(\rho_{A_1A_2}))$ in general, and  the equalities hold for the special cases of pure states [22]
\begin{lemma}\label{lemma:1}

(1) If $t\geq k^\omega\geq k\geq1$, $\omega\geq1$, and $0\leq x\leq\frac{1}{2}$, we have
\begin{eqnarray}
(1+t)^x\geq(\frac{1}{2})^x+\frac{(1+k^\omega)^x-(\frac{1}{2})^x}{k^{\omega x}}t^x.
\end{eqnarray}

(2) If $0\leq t\leq k^\omega\leq k\leq1$, $\omega\geq1$, and $x\geq 1$, we have
\begin{eqnarray}
(1+t)^x\geq(\frac{1}{2})^x+\frac{(1+k^\omega)^x-(\frac{1}{2})^x}{k^{\omega x}}t^x.
\end{eqnarray}
\end{lemma}

\begin{proof} We prove these two inequalities 
in a similar manner.
Consider  $h(x,y)=(1+\frac{1}{y})^{x-1}-(\frac{1}{2})^x$ where $0\leq x\leq\frac{1}{2}$ and $0<y\leq\frac{1}{k^\omega}$ with real numbers $k\geq1$ and $\omega\geq1$. Then $\frac{\partial h}{\partial x}=(1+\frac{1}{y})^{x-1}ln(1+\frac{1}{y})-(\frac{1}{2})^xln\frac{1}{2}>0$ as $1+\frac{1}{y}\geq2$. So $h(x,y)$ is an increasing function of $x$ when $y$ is fixed, i.e, $h(x,y)\leq h(\frac{1}{2},y)=(1+\frac{1}{y})^{-\frac{1}{2}}-(\frac{1}{2})^{\frac{1}{2}}\leq0$ as $0<(1+\frac{1}{y})^{-1}\leq\frac{1}{2}$. Let $g(x,y)=(1+y)^x-(\frac{1}{2}y)^x$ with $0\leq x\leq\frac{1}{2}$ and $0<y\leq\frac{1}{k^\omega}$. We have $\frac{\partial g}{\partial y}=xy^{x-1}[(1+\frac{1}{y})^{x-1}-(\frac{1}{2})^x]\leq0$ and then $g(x,y)$ is a decreasing function of $y$ for fixed $x$. 
Thus, for $t\geq k^\omega$, one has $g(x,\frac{1}{t})\geq g(x, \frac1{k^\omega})$. Therefore $(1+t)^x\geq(\frac{1}{2})^x+\frac{(1+k^\omega)^x-(\frac{1}{2})^x}{k^{\omega x}}t^x$.
\end{proof}

\begin{lemma}\label{lemma:2}
For $0\leq x,y, x^2+y^2\leq1$, $0\leq\alpha\leq\frac{\eta}{2}$, $\eta\geq1$, $\omega\geq1$, and $k\geq1$,

(1)if $f_B^\eta(y)\geq k^\omega f_B^\eta(x)$, we have
\begin{eqnarray}
f_B^\alpha(\sqrt{x^2+y^2})\geq(\frac{1}{2})^{\frac{\alpha}{\eta}}f_B^\alpha(x)+\frac{(1+k^\omega)^{\frac{\alpha}{\eta}}-(\frac{1}{2})^{\frac{\alpha}{\eta}}}{k^{\frac{\omega\alpha}{\eta}}}f_B^\alpha(y).
\end{eqnarray}

(2)if $f_G^\eta(y)\geq k^\omega f_G^\eta(x)$, we have
\begin{eqnarray}
f_G^\alpha(\sqrt{x^2+y^2})\geq(\frac{1}{2})^{\frac{\alpha}{\eta}}f_G^\alpha(x)+\frac{(1+k^\omega)^{\frac{\alpha}{\eta}}-(\frac{1}{2})^{\frac{\alpha}{\eta}}}{k^{\frac{\omega\alpha}{\eta}}}f_G^\alpha(y).
\end{eqnarray}

\end{lemma}
\begin{proof}
When $f_B^\eta(y)\geq k^\omega f_B^\eta(x)$, we have
\begin{eqnarray}
f_B^\alpha(\sqrt{x^2+y^2})&\geq&(f_B^\eta(x)+f_B^\eta(y))^\frac\alpha \eta\nonumber\\
&=&f_B^\alpha(x)(1+\frac{f_B^\eta(y)}{f_B^\eta(x)})^\frac\alpha \eta\nonumber\\
&\geq& f_B^\alpha(x)[(\frac{1}{2})^\frac\alpha \eta+\frac{(1+k^\omega)^\frac\alpha\eta-(\frac{1}{2})^\frac\alpha \eta}{k^{\frac{\omega\alpha} \eta}}(\frac{f_B^\eta(y)}{f_B^\eta(x)})^\frac\alpha \eta]\nonumber\\
&=&(\frac{1}{2})^{\frac{\alpha}{\eta}}f_B^\alpha(x)+\frac{(1+k^\omega)^{\frac{\alpha}{\eta}}-(\frac{1}{2})^{\frac{\alpha}{\eta}}}{k^{\frac{\omega\alpha}{\eta}}}f_B^\alpha(y),
\end{eqnarray}
where $0\leq\alpha\leq\frac{\eta}{2}$, $\eta\geq1$, $\omega\geq1$, and $k\geq1$. The first inequality is obtained by $f_B^\eta(\sqrt{x^2+y^2})\geq f_B^\eta(x)+f_B^\eta(y)$ for $0\leq x,y, x^2+y^2\leq1$ and $\eta\geq1$ [21] and the second one is due to (6) of Lemma 1.
\end{proof}

\begin{thm}\label{thm:1}
In tri-qubit quantum systems, assuming real numbers $k\geq1$, $\omega\geq1$, $0\leq\alpha\leq\frac{\eta}{2}$ and $\eta\geq1$, then one has that

(1) if $M_B^\eta(\rho_{A_1A_3})\geq k^\omega M_B^\eta(\rho_{A_1A_2})$, then the Bures measure of entanglement satisfies
\begin{eqnarray}
M_B^\alpha(\rho_{A_1|A_2A_3})\geq(\frac{1}{2})^{\frac{\alpha}{\eta}}M_B^\alpha(\rho_{A_1A_2})+\frac{(1+k^\omega)^{\frac{\alpha}{\eta}}-(\frac{1}{2})^{\frac{\alpha}{\eta}}}{k^{\frac{\omega\alpha}{\eta}}}M_B^\alpha(\rho_{A_1A_3}).
\end{eqnarray}

(2) if $M_B^\eta(\rho_{A_1A_2})\geq k^\omega M_B^\eta(\rho_{A_1A_3})$, then the Bures measure of entanglement satisfies
\begin{eqnarray}
M_B^\alpha(\rho_{A_1|A_2A_3})\geq(\frac{1}{2})^{\frac{\alpha}{\eta}}M_B^\alpha(\rho_{A_1A_3})+\frac{(1+k^\omega)^{\frac{\alpha}{\eta}}-(\frac{1}{2})^{\frac{\alpha}{\eta}}}{k^{\frac{\omega\alpha}{\eta}}}M_B^\alpha(\rho_{A_1A_2}).
\end{eqnarray}
\end{thm}

\begin{proof}
For an arbitrary tri-qubit state $\rho$ under bipartite partition $A_1|A_2A_3$, one has [26]:
\begin{eqnarray}
C^2(\rho_{A_1|A_2A_3})\geq C^2(\rho_{A_1A_2})+C^2(\rho_{A_1A_3}).
\end{eqnarray}
Suppose $M_B^\eta(\rho_{A_1A_3})\geq k^\omega M_B^\eta(\rho_{A_1A_2})$, $k\geq1$ and $\omega\geq1$, then
\begin{eqnarray}
M_B^{\alpha}(\rho_{A_1|A_2A_3})&\geq&f_B^\alpha(C(\rho_{A_1|A_2A_3}))\nonumber\\
&\geq&f_B^\alpha(\sqrt{C^2(\rho_{A_1A_2})+C^2(\rho_{A_1A_3})})\nonumber\\
&\geq&(\frac{1}{2})^{\frac{\alpha}{\eta}}f_B^{\alpha}(C(\rho_{A_1A_2}))+\frac{(1+k^\omega)^{\frac{\alpha}{\eta}}-(\frac{1}{2})^{\frac{\alpha}{\eta}}}{k^{\frac{\omega\alpha}{\eta}}}f_B^{\alpha}(C(\rho_{A_1A_3})) \nonumber\\
&=&(\frac{1}{2})^{\frac{\alpha}{\eta}}M_B^{\alpha}(\rho_{A_1A_2})+\frac{(1+k^\omega)^{\frac{\alpha}{\eta}}-(\frac{1}{2})^{\frac{\alpha}{\eta}}}{k^{\frac{\omega\alpha}{\eta}}}M_B^{\alpha}(\rho_{A_1A_3}),
\end{eqnarray}
where $0\leq\alpha\leq\frac{\eta}{2}$, $\eta\geq1$, the first inequality is obtained by $M_B(\rho_{A_1|A_2A_3})\geq f_B(C(\rho_{A_1|A_2A_3}))$ [22], the second one is due to inequality (13) and the fact that $f_B(x)$ is a monotonically increasing function, and the last inequality is due to Lemma 2. The equality holds since $M_B(\rho)=f_B(C(\rho))$ for 2-qubit states [22]. Similar proof gives inequality (12) by using Lemma 2.
\end{proof}
{\bf Remark 1.} We have derived the monogamy relations for the Bures measure of entanglement and also for the geometric measure of entanglement by the same argument.

In the following, let $M(\rho_{A_1A_i})=M_{A_1A_i}$, $C(\rho_{A_1A_i})=C_{A_1A_i}$, $M(\rho_{A_1|A_{j+1}\cdots A_n})=M_{A_1|A_{j+1}\cdots A_n}$, $C(\rho_{A_1|A_{j+1}\cdots A_n})=C_{A_1|A_{j+1}\cdots A_n}$ where $i=2,\cdots,n-1$ and $j=1,\cdots,n-1$ and simply note the Bures measure of entanglement ($M_B$) or the geometric
measure of entanglement ($M_G$) by $M$. We now generalize the monogamy inequalities of the $\alpha$th ($0\leq\alpha\leq\frac{\eta}{2}$, $\eta\geq1$) power of the Bures measure of entanglement for $n$-qubit quantum states $\rho$ under bipartite partition $A_1|A_2\cdots A_n$.
\begin{thm}\label{thm:3}
In multi-qubit quantum systems, assuming real numbers $k\geq1$, $0\leq\alpha\leq\frac{\eta}{2}$, $\eta\geq1$, $\omega\geq1$, and $\mu=\frac{(1+k^\omega)^{\frac{\alpha}{\eta}}-(\frac{1}{2})^{\frac{\alpha}{\eta}}}{k^{\frac{\omega\alpha}{\eta}}}$, we have that

(1) if $k^\omega M_{A_1A_i}^\eta\leq M_{A_1|A_{i+1}\cdots A_n}^\eta$ for $i=2,\cdots,m$ and $M_{A_1A_j}^\eta\geq k^\omega M_{A_1|A_{j+1}\cdots A_n}^\eta$ for $j=m+1,\cdots,n-1$, $\forall$ $2\leq m\leq n-2$, $n\geq4$, then we have
\begin{eqnarray}
M_{A_1|A_2A_3\cdots A_n}^\alpha&\geq&(\frac{1}{2})^{\frac{\alpha}{\eta}}(M_{A_1A_2}^\alpha+\mu M_{A_1A_3}^\alpha+\cdots+\mu^{m-2}E_{A_1A_m}^\alpha)\nonumber\\
&+&\mu^m[M_{A_1A_{m+1}}^\alpha+(\frac{1}{2})^{\frac{\alpha}{\eta}}M_{A_1A_{m+2}}^\alpha+\cdots+(\frac{1}{2})^{\frac{(n-m-2)\alpha}{\eta}}M_{A_1A_{n-1}}^\alpha]\nonumber\\
&+&\mu^{m-1}(\frac{1}{2})^{\frac{(n-m-1)\alpha}{\eta}}M_{A_1A_n}^\alpha.
\end{eqnarray}

(2) if $k^\omega M_{A_1A_i}^\eta\leq M_{A_1|A_{i+1}\cdots A_n}^\eta$ for $i=2,\cdots,n-1$ and $n\geq3$, then we have that
\begin{eqnarray}
M_{A_1|A_2A_3\cdots A_n}^\alpha\geq(\frac{1}{2})^{\frac{\alpha}{\eta}}(M_{A_1A_2}^\alpha+\mu M_{A_1A_3}^\alpha+\cdots+\mu^{n-3}M_{A_1A_{n-1}}^\alpha)+\mu^{n-2}M_{A_1A_n}^\alpha.
\end{eqnarray}

(3) if $M_{A_1A_i}^\eta\geq k^\omega M_{A_1|A_{i+1}\cdots A_n}^\eta$ for $i=2,\cdots,n-1$ and $n\geq3$, then we have that
\begin{eqnarray}
M_{A_1|A_2A_3\cdots A_n}^\alpha\geq \mu(M_{A_1A_2}^\alpha+(\frac{1}{2})^{\frac{\alpha}{\eta}}M_{A_1A_3}^\alpha+\cdots+(\frac{1}{2})^{\frac{(n-3)\alpha}{\eta}}M_{A_1A_{n-1}}^\alpha)+(\frac{1}{2})^{\frac{(n-2)\alpha}{\eta}}M_{A_1A_n}^\alpha.
\end{eqnarray}
\end{thm}

\begin{proof}

For an $n$-qubit quantum state $\rho$ under bipartite partition $A_1|A_2A_3\cdots A_n$, if $k^\omega M_{A_1A_i}^\eta\leq M_{A_1|A_{i+1}\cdots A_n}^\eta$ for $i=2,\cdots,m$, we have
\begin{eqnarray}
M_{A_1|A_2A_3\cdots A_n}^\alpha&\geq& f^\alpha(C_{A_1|A_2A_3\cdots A_n})\nonumber\\
&\geq&f^\alpha(\sqrt{C_{A_1A_2}^2+C_{A_1|A_3\cdots A_n}^2}) \nonumber\\
&\geq&(\frac{1}{2})^{\frac{\alpha}{\eta}}f^\alpha(C_{A_1A_2})+\mu f^\alpha(C_{A_1|A_3\cdots A_n}) \nonumber\\
&\geq&\cdots\\
&\geq&(\frac{1}{2})^{\frac{\alpha}{\eta}}(f^\alpha(C_{A_1A_2})+\mu f^\alpha(C_{A_1A_3})+\cdots+\mu^{m-2}f^\alpha(C_{A_1A_m}))\nonumber\\
&+&\mu^{m-1}f^\alpha(C_{A_1|A_{m+1}\cdots A_n})\nonumber\\
&=&(\frac{1}{2})^{\frac{\alpha}{\eta}}(M_{A_1A_2}^\alpha+\mu M_{A_1A_3}^\alpha+\cdots+\mu^{m-2}M_{A_1A_m}^\alpha)+\mu^{m-1}f^\alpha(C_{A_1|A_{m+1}\cdots A_n}),\nonumber
\end{eqnarray}
where the first inequality follows from $M_{A_1|A_2A_3\cdots A_n}\geq f(C_{A_1|A_2A_3\cdots A_n})$ [22], the second one is due to $C^2_{A_1|A_2A_3}\geq C^2_{A_1A_2}+C^2_{A_1A_3}$ for $2\otimes2\otimes2^{n-2}$ tripartite state [26], and $f(x)$ being a monotonically increasing function. Using Lemma 2, we get the third inequality. Other inequalities are consequences of Lemma 2
and the last equality holds due to $M(\rho)=f(C(\rho))$ for 2-qubit states.

For $M_{A_1A_j}^\eta\geq k^\omega M_{A_1|A_{j+1}\cdots A_n}^\eta$ for $j=m+1,\cdots,n-1$, similar argument gives the following inequality by using Lemma 2:
\begin{eqnarray}
f^\alpha(C_{A_1|A_{m+1}+\cdots+A_n})&\geq& \mu f^\alpha(C_{A_1A_{m+1}})+(\frac{1}{2})^{\frac{\alpha}{\eta}}f^\alpha(C_{A_1|A_{m+2}\cdots A_n})\nonumber\\
&\geq&\cdots \nonumber\\
&\geq&\mu[M_{A_1A_{m+1}}^\alpha+(\frac{1}{2})^{\frac{\alpha}{\eta}}M_{A_1A_{m+2}}^\alpha+\cdots+(\frac{1}{2})^{\frac{(n-m-2)\alpha}{\eta}}M_{A_1A_{n-1}}^\alpha]\nonumber\\
&+&(\frac{1}{2})^{\frac{(n-m-1)\alpha}{\eta}}M_{A_1A_n}^\alpha.
\end{eqnarray}
Combining (18) and (19), one obtains (15). If all $k^\omega M_{A_1A_i}^\eta\leq M_{A_1|A_{i+1}\cdots A_n}^\eta$ for $i=2,\cdots,n-1$ or $M_{A_1A_i}^\eta\geq k^\omega M_{A_1|A_{i+1}\cdots A_n}^\eta$ for $i=2,\cdots,n-1$, we have the inequalities (16) and (17).
\end{proof}

{\bf Remark 2.} We use the Bures measure of entanglement as an example to compare our result with those in [22, 27, 28]. In tripartite quantum systems, when $M_B^\eta(\rho_{A_1A_3})\geq k^\omega M_B^\eta(\rho_{A_1A_2})\geq k M_B^\eta(\rho_{A_1A_2})$ for $k\geq1$, $\omega\geq1$, $0\leq\alpha\leq\frac{\eta}{2}$ and $\eta\geq1$,  Theorem 1 says that the $\alpha$th power of the  Bures measure of entanglement satisfies $M_B^\alpha(\rho_{A_1|A_2A_3})\geq(\frac{1}{2})^{\frac{\alpha}{\eta}}M_B^\alpha(\rho_{A_1A_2})+\frac{(1+k^\omega)^{\frac{\alpha}{\eta}}-(\frac{1}{2})^{\frac{\alpha}{\eta}}}{k^{\frac{\omega\alpha}{\eta}}}M_B^\alpha(\rho_{A_1A_3})$ denoted as $m$. On the other hand, the lower bounds of $M_B^\alpha(\rho_{A_1|A_2A_3})$ are $M_B^\alpha(\rho_{A_1A_2})+M_B^\alpha(\rho_{A_1A_3})\doteqdot m_1$ by [22], $(\frac{1}{2})^{\frac{\alpha}{\eta}}M_B^\alpha(\rho_{A_1A_2})+\frac{(1+k)^{\frac{\alpha}{\eta}}-(\frac{1}{2})^{\frac{\alpha}{\eta}}}{k^{\frac{\alpha}{\eta}}}M_B^\alpha(\rho_{A_1A_3})\doteqdot m_2$ by [28], and $M_B^\alpha(\rho_{A_1A_2})+\frac{(1+k^\omega)^{\frac{\alpha}{\eta}}-1}{k^{\frac{\omega\alpha}{\eta}}}M_B^\alpha(\rho_{A_1A_3})\doteqdot m_3$ by [27] respectively.
Let $\Delta m=m-m_i$, $i=1, 2, 3$, in the following examples we will see that $\Delta m\geq0$, so our results are tighter than those in [22, 27, 28] for $0\leq\alpha\leq\frac{\eta}{2}$ and $\eta\geq1$.





\textit{\textbf{Example 1.}} Let us consider the 3-qubit quantum state of generalized Schmidt decomposition $|\varphi\rangle$,
\begin{eqnarray}
|\varphi\rangle=\lambda_0|000\rangle+\lambda_1e^{i\theta}|100\rangle+\lambda_2|101\rangle+\lambda_3|110\rangle+\lambda_4|111\rangle,
\end{eqnarray}
where $0\leq\theta\leq\pi$, $\lambda_i\geq0$, $i=0,\cdots,4$ and $\sum_{i=0}^4\lambda_i^2=1$.
One computes that, one has $C(|\varphi\rangle_{A_1|A_2A_3})=2\lambda_0\sqrt{\lambda_2^2+\lambda_3^2+\lambda_4^2}$, $C(|\varphi\rangle_{A_1A_2})=2\lambda_0\lambda_2$, and $C(|\varphi\rangle_{A_1A_3})=2\lambda_0\lambda_3$. Let $\lambda_0=\lambda_3=\frac{\sqrt{2}}{3}$, $\lambda_2=\frac{\sqrt{5}}{3}$, $\lambda_1=\lambda_4=0$, $k=2$, and $\omega=1.5$,  then we have $M_B(|\varphi\rangle_{A_1|A_2A_3})\approx0.23617$, $M_B(|\varphi\rangle_{A_1A_2})\approx0.14989$, $M_B(|\varphi\rangle_{A_1A_3})\approx0.05279$. Therefore, $M_B^x(|\varphi\rangle_{A_1|A_2A_3})=0.23617^x$. By Theorem 1, the lower bound of $M_B^x(|\varphi\rangle_{A_1|A_2A_3})$ is $z_1=(\frac{1}{2})^{\frac{x}{y}}0.05279^x+\frac{3.82842^{\frac{x}{y}}-(\frac{1}{2})^{\frac{x}{y}}}{2.82842^{\frac{x}{y}}}0.14989^x$. By Theorem 1 in [22], the lower bound of $M_B^x(|\varphi\rangle_{A_1|A_2A_3})$ is $z_2=0.05279^x+0.14989^x$. Assuming $0\leq x\leq\frac{1}{2}$ and $1\leq y\leq10$, Fig. 1 verifies that our result is tighter than that of [22].

\begin{figure}[!htb]
\centerline{\includegraphics[width=0.6\textwidth]{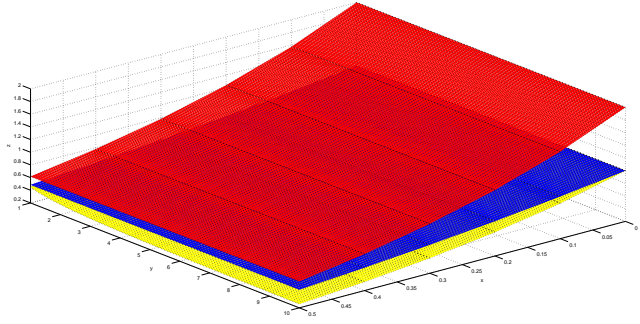}}
\renewcommand{\figurename}{Fig.}
\caption{
The red surface represents the Bures measure of entanglement of the state $|\varphi\rangle_{A_1|A_2A_3}$. The lower bound in [22] is shown by the yellow surface and the blue surface is our result in Theorem 1.}
\end{figure}

\textit{\textbf{Example 2.}} Consider the 3-qubit generalized W-class state $\rho=|\psi\rangle_{A_1A_2A_3}\langle\psi|$,
\begin{eqnarray}
|\psi\rangle_{A_1A_2A_3}=\frac{1}{\sqrt{6}}|100\rangle+\frac{1}{\sqrt{6}}|010\rangle+\frac{2}{\sqrt{6}}|001\rangle.
\end{eqnarray}
By the definition of concurrence, we have $C(\rho_{A_1|A_2A_3})=\frac{\sqrt{5}}{3}$, $C(\rho_{A_1A_2})=\frac{1}{3}$, and $C(\rho_{A_1A_3})=\frac{2}{3}$. Thus $M_B(\rho_{A_1|A_2A_3})=2-2\sqrt\frac56\approx0.17426$, $M_B(\rho_{A_1A_2})=2-2\sqrt\frac{3+2\sqrt2}6\approx0.02880$, and $M_B(\rho_{A_1A_3})=2-2\sqrt\frac{3+2\sqrt5}6\approx0.13166$. Let $k=2$, $\omega=2$, and $\eta=2$, then by Theorem 1, the lower bound of $M_B^x(\rho_{A_1|A_2A_3})$ is $y_1=(\frac{1}{2})^{\frac{x}{2}}0.02880^x+\frac{(1+4)^{\frac{x}{2}}-(\frac{1}{2})^{\frac{x}{2}}}{4^{\frac{x}{2}}}0.13166^x$. Using the method in [27] and [28], we have that the lower bound of $M_B^x(\rho_{A_1|A_2A_3})$ is $y_2=0.02880^x+\frac{(1+4)^{\frac{x}{2}}-1}{4^{\frac{x}{2}}}0.13166^x$ and $y_3=(\frac{1}{2})^{\frac{x}{2}}0.02880^x+\frac{(1+2)^{\frac{x}{2}}-(\frac{1}{2})^{\frac{x}{2}}}{2^{\frac{x}{2}}}0.13166^x$ respectively. Fig. 2 depicts the value of $y$ for $0\leq x\leq\frac{1}{2}$, which shows that
Theorem 1 supplies a better estimation of the Bures measure of entanglement than those of [27] and [28].

\begin{figure}[!htb]
\centerline{\includegraphics[width=0.6\textwidth]{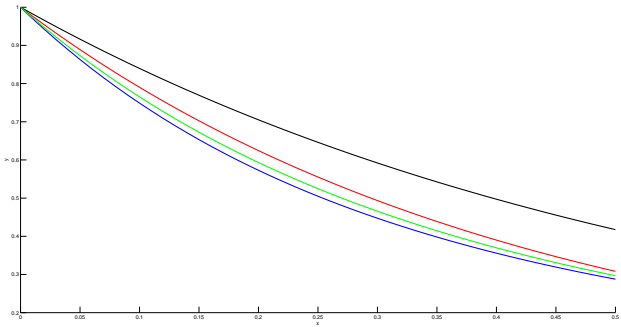}}
\renewcommand{\figurename}{Fig.}
\caption{
The black line represents the Bures measure of entanglement of the state $|\psi\rangle$. The red, green and blue line represent the lower bound $y_1$, $y_2$, and $y_3$  respectively.}
\end{figure}

\section{Monogamy relations of concurrence in higher-dimensional quantum systems}
In the above section, we have given the monogamy inequalities of the Bures measure of entanglement and the geometric measure of entanglement, both of which can be expressed as functions of concurrence for qubit quantum states. In the following, we will present the monogamy relations of concurrence for qudit quantum states and present a method for higher-dimensional monogamy relations. We first introduce the following definition.

Let $H_{A_1}$ and $H_{A_2}$ be $d_{A_1}$- and $d_{A_2}$-dimensional Hilbert spaces. The concurrence of a bipartite quantum pure state $|\varphi\rangle_{A_1A_2}\in H_{A_1}\otimes\ H_{A_2}$ is defined by [29],
\begin{eqnarray}
C(|\varphi\rangle_{A_1A_2})=\sqrt{2[1-tr(\rho_{A_1}^2)]},
\end{eqnarray}
where $\rho_{A_1}$ is the reduced density matrix of $\rho=|\varphi\rangle_{A_1A_2}\langle\varphi|$, i.e, $\rho_{A_1}=tr_{A_2}(\rho)$. For a mixed bipartite quantum state $\rho_{A_1A_2}=\sum_ip_i|\varphi_i\rangle_{A_1A_2}\langle\varphi_i|\in H_{A_1}\otimes H_{A_2}$, the concurrence is given by the convex roof
\begin{eqnarray}
C(\rho_{A_1A_2})=\min_{\{p_i,|\varphi_i\rangle\}}\sum_ip_iC(|\varphi_i\rangle_{A_1A_2}),
\end{eqnarray}
where the minimum is taken over all possible convex partitions of $\rho_{A_1A_2}$ into pure state ensembles $\{p_i,|\varphi_i\rangle\}$, $0\leq p_i\leq1$ and $\sum_ip_i=1$.

Now consider the concurrence for a tripartite quantum state under bipartite partition $A_1|A_2A_3$. For a pure tripartite quantum state $|\varphi\rangle_{A_1|A_2A_3}^{d_1\otimes d_2\otimes d_3}\in\mathcal{H}_{A_1}^{d_1}\otimes\mathcal{H}_{A_2}^{d_2}\otimes\mathcal{H}_{A_3}^{d_3}$, it has the form as follows:
\begin{eqnarray}
|\varphi\rangle_{A_1|A_2A_3}^{d_1\otimes d_2\otimes d_3}=\sum_{a=1}^{d_1}\sum_{b=1}^{d_2}\sum_{c=1}^{d_3}h_{a|bc}|abc\rangle,
\end{eqnarray}
where $h_{a|bc}\in$ $\mathbb{C}$, $\sum_{abc}h_{a|bc}h_{a|bc}^\ast=1$. From the definition of concurrence in (22), the concurrence of $|\varphi\rangle_{A_1|A_2A_3}^{d_1\otimes d_2\otimes d_3}$ is given by:
\begin{eqnarray}
C^2(|\varphi\rangle_{A_1|A_2A_3}^{d_1\otimes d_2\otimes d_3})=\sum_{a,e=1}^{d_1}\sum_{b,f=1}^{d_2}\sum_{c,j=1}^{d_3}|h_{a|bc}h_{e|fj}-h_{e|bc}h_{a|fj}|^2.
\end{eqnarray}

Next we analyse a pure substate of $|\varphi\rangle_{A_1|A_2A_3}^{d_1\otimes d_2\otimes d_3}$, i.e, $|\varphi\rangle_{A_1|A_2A_3}^{2\otimes 2\otimes 2}=\sum_{a\in\{a_1,a_2\}}\sum_{b\in\{b_1,b_2\}}$$\sum_{c\in\{c_1,c_2\}}$\\$h_{a|bc}|abc\rangle$, where $a_1\neq a_2\in\{1,2, \cdots, d_1\}$,
$b_1\neq b_2\in\{1,2, \cdots, d_2\}$ and
$c_1\neq c_2\in\{1,2, \cdots, d_3\}$. There are $d_1\choose 2$$d_2\choose 2$$d_3\choose 2$ different substates, where $d_i\choose 2$, $i=1, 2, 3$, is the binomial coefficient, and we simply use $|\varphi\rangle_{A_1|A_2A_3}^{2\otimes2\otimes2}$ to denote one of the substates.
It follows from (25) that
\begin{eqnarray}
C^2(|\varphi\rangle_{A_1|A_2A_3}^{d_1\otimes d_2\otimes d_3})\geq\sum\frac{1}{(d_1-1)(d_2-1)(d_3-1)}C^2(|\varphi\rangle_{A_1|A_2A_3}^{2\otimes2\otimes2}),
\end{eqnarray}
where $\sum$ stands for summing over all possible pure sub-states $|\varphi\rangle_{A_1|A_2A_3}^{2\otimes2\otimes2}$.
For a mixed state $\rho_{A_1|A_2A_3}^{d_1\otimes d_2\otimes d_3}$, its substate $\rho_{A_1|A_2A_3}^{2\otimes2\otimes2}$ has the following form,
\begin{eqnarray}
\rho_{A_1|A_2A_3}^{2\otimes2\otimes2}=
\left[ \begin{array}{cccccccc}
           \rho_{a_1|b_1c_1},_{a_1|b_1c_1}&\rho_{a_1|b_1c_1},_{a_1|b_1c_2}&\cdots&\rho_{a_1|b_1c_1},_{a_2|b_2c_1}&\rho_{a_1|b_1c_1},_{a_2|b_2c_2}  \\
           \rho_{a_1|b_1c_2},_{a_1|b_1c_1}&\rho_{a_1|b_1c_2},_{a_1|b_1c_2}&\cdots&\rho_{a_1|b_1c_2},_{a_2|b_2c_1}&\rho_{a_1|b_1c_2},_{a_2|b_2c_2}  \\
           \rho_{a_1|b_2c_1},_{a_1|b_1c_1}&\rho_{a_1|b_2c_1},_{a_1|b_1c_2}&\cdots&\rho_{a_1|b_2c_1},_{a_2|b_2c_1}&\rho_{a_1|b_2c_1},_{a_2|b_2c_2}  \\
           \vdots&\vdots&\vdots&\vdots&\vdots   \\
           \rho_{a_2|b_1c_2},_{a_1|b_1c_1}&\rho_{a_2|b_1c_2},_{a_1|b_1c_2}&\cdots&\rho_{a_2|b_1c_2},_{a_2|b_2c_1}&\rho_{a_2|b_1c_2},_{a_2|b_2c_2}  \\
           \rho_{a_2|b_2c_1},_{a_1|b_1c_1}&\rho_{a_2|b_2c_1},_{a_1|b_1c_2}&\cdots&\rho_{a_2|b_2c_1},_{a_2|b_2c_1}&\rho_{a_2|b_2c_1},_{a_2|b_2c_2}  \\
           \rho_{a_2|b_2c_2},_{a_1|b_1c_1}&\rho_{a_2|b_2c_2},_{a_1|b_1c_2}&\cdots&\rho_{a_2|b_2c_2},_{a_2|b_2c_1}&\rho_{a_2|b_2c_2},_{a_2|b_2c_2}  \\
           \end{array}
      \right ],
\label{A}
\end{eqnarray}
which is an unnormalized tripartite mixed state.
\begin{lemma}
In tripartite quantum systems $\mathcal{H}^{d_1}_{A_1} \otimes\mathcal{H}^{d_2}_{A_2}\otimes\mathcal{H}^{d_3}_{A_3}$, the concurrence of mixed state $\rho_{A_1|A_2A_3}^{d_1\otimes d_2\otimes d_3}$ satisfies:
\begin{eqnarray}
C^2(\rho_{A_1|A_2A_3}^{d_1\otimes d_2\otimes d_3})\geq\sum\frac{1}{(d_1-1)(d_2-1)(d_3-1)}C^2(\rho_{A_1|A_2A_3}^{2\otimes2\otimes2}),
\end{eqnarray}
where $\sum$ stands for summing over all possible mixed substates $\rho_{A_1|A_2A_3}^{2\otimes2\otimes2}$.
\end{lemma}
\begin{proof}
For a mixed state $\rho_{A_1|A_2A_3}^{d_1\otimes d_2\otimes d_3}=\sum_i p_i|\varphi_i\rangle_{A_1|A_2A_3}^{d_1\otimes d_2\otimes d_3}\langle\varphi_i|$, we have
\begin{eqnarray}
C(\rho_{A_1|A_2A_3}^{d_1\otimes d_2\otimes d_3})&=&min \sum_ip_iC(|\varphi_i\rangle_{A_1|A_2A_3}^{d_1\otimes d_2\otimes d_3})\nonumber\\
&\geq& min \frac{1}{\sqrt{(d_1-1)(d_2-1)(d_3-1)}}\sum_ip_i(\sum C^2(|\varphi_i\rangle_{A_1|A_2A_3}^{2\otimes2\otimes2}))^\frac{1}{2}\nonumber\\
&\geq& min
\frac{1}{\sqrt{(d_1-1)(d_2-1)(d_3-1)}}[\sum(\sum_ip_iC(|\varphi_i\rangle_{A_1|A_2A_3}^{2\otimes2\otimes2}))^2]^\frac{1}{2}   \nonumber \\
&\geq& \frac{1}{\sqrt{(d_1-1)(d_2-1)(d_3-1)}}[\sum(min \sum_ip_iC(|\varphi_i\rangle_{A_1|A_2A_3}^{2\otimes2\otimes2}))^2]^\frac{1}{2}    \nonumber\\
&=&\frac{1}{\sqrt{(d_1-1)(d_2-1)(d_3-1)}}[\sum C^2(\rho_{A_1|A_2A_3}^{2\otimes2\otimes2})]^\frac{1}{2},
\end{eqnarray}
where we have used the Minkowski inequality $\sum_m(\sum_n a_{mn}^2)^\frac{1}{2}\geq(\sum_n(\sum_ma_{mn})^2)^\frac{1}{2}$ in the second inequality, the minimum is taken over all possible pure state decompositions of mixed state $\rho_{A_1|A_2A_3}^{d_1\otimes d_2\otimes d_3}$ in the first three minimizations, while the minimum in the last inequality is taken over all pure state decompositions of $\rho_{A_1|A_2A_3}^{2\otimes2\otimes2}$.
\end{proof}
By using $(1+t)^x\geq 1+t^x$ for $x\geq1$ and $0\leq t\leq1$, we can easily obtain $(\sum_{i=1}^nc_i)^x\geq\sum_{i=1}^nc_i^x$ for nonnegativity numbers $c_i$, $i=1, \cdots, n$, and $x\geq1$. The monogamy inequalities of the $\alpha$th ($\alpha\geq2$) power of the concurrence for $n$-qudit quantum states are given as follows.
\begin{thm}\label{thm:1}
In tripartite quantum systems $\mathcal{H}^{d_1} \otimes\mathcal{H}^{d_2}\otimes\mathcal{H}^{d_3}$, assuming $0\leq k^\omega\leq k\leq1$, $\omega\geq1$, and $\alpha\geq2$,

(1) if $C^2(\rho_{A_1A_3}^{2\otimes2})\leq k^\omega C^2(\rho_{A_1A_2}^{2\otimes2})$, the concurrence satisfies
\begin{eqnarray}
C^{\alpha}(\rho_{A_1|A_2A_3}^{d_1\otimes d_2\otimes d_3})\geq &\sum&[\frac1{(d_1-1)(d_2-1)(d_3-1)}]^\frac\alpha2[(\frac12)^{\frac{\alpha}{2}}C^{\alpha}(\rho_{A_1A_2}^{2\otimes2})\nonumber\\
&+&\frac{(1+k^\omega)^{\frac{\alpha}{2}}-(\frac12)^{\frac{\alpha}{2}}}{k^{\frac{\omega\alpha}{2}}}C^{\alpha}(\rho_{A_1A_3}^{2\otimes2})].
\end{eqnarray}

(2) if $C^2(\rho_{A_1A_2}^{2\otimes2})\leq k^\omega C^2(\rho_{A_1A_3}^{2\otimes2})$, the concurrence satisfies
\begin{eqnarray}
C^{\alpha}(\rho_{A_1|A_2A_3}^{d_1\otimes d_2\otimes d_3})\geq &\sum&[\frac1{(d_1-1)(d_2-1)(d_3-1)}]^\frac\alpha2[(\frac12)^{\frac{\alpha}{2}}C^{\alpha}(\rho_{A_1A_3}^{2\otimes2})\nonumber\\
&+&\frac{(1+k^\omega)^{\frac{\alpha}{2}}-(\frac12)^{\frac{\alpha}{2}}}{k^{\frac{\omega\alpha}{2}}}C^{\alpha}(\rho_{A_1A_2}^{2\otimes2})].
\end{eqnarray}
\end{thm}
\begin{proof}
For $\alpha\geq2$, we have
\begin{eqnarray}
C^{\alpha}(\rho_{A_1|A_2A_3}^{d_1\otimes d_2\otimes d_3})&\geq&[\sum\frac{1}{(d_1-1)(d_2-1)(d_3-1)}C^2(\rho_{A_1|A_2A_3}^{2\otimes2\otimes2})]^\frac\alpha2\nonumber\\
&=&[\frac1{(d_1-1)(d_2-1)(d_3-1)}]^\frac\alpha2(\sum C^2(\rho_{A_1|A_2A_3}^{2\otimes2\otimes2}))^\frac\alpha2\nonumber\\
&\geq&\sum[\frac1{(d_1-1)(d_2-1)(d_3-1)}]^\frac\alpha2(C^2(\rho_{A_1|A_2A_3}^{2\otimes2\otimes2}))^\frac\alpha2,
\end{eqnarray}
where the first inequality follows from Lemma 3, the second inequality is due to $(\sum_{i=1}^nc_i)^x\geq\sum_{i=1}^nc_i^x$ for nonnegativity numbers $c_i$, $i=1, \cdots, n$, and $x\geq1$, and $\sum$ stands for summing over all possible $2\otimes 2\otimes 2$ mixed substates of $\rho_{A_1|A_2A_3}^{d_1\otimes d_2\otimes d_3}$.
Assuming $C^2(\rho_{A_1A_3}^{2\otimes2})\leq k^\omega C^2(\rho_{A_1A_2}^{2\otimes2})$, by using $C^2(\rho_{A_1|A_2A_3})\geq C^2(\rho_{A_1A_2})+C^2(\rho_{A_1A_3})$ [26] and inequality (7) of Lemma 1, one has
\begin{eqnarray}
C^\alpha(\rho_{A_1|A_2A_3}^{2\otimes2\otimes2})
&\geq&(C^2(\rho_{A_1A_2}^{2\otimes2})+C^2(\rho_{A_1A_3}^{2\otimes2}))^\frac\alpha2\nonumber\\
&=&C^{\alpha}(\rho_{A_1A_2}^{2\otimes2})(1+\frac{C^2(\rho_{A_1A_3}^{2\otimes2})}{C^2(\rho_{A_1A_2}^{2\otimes2})})^\frac\alpha2\nonumber\\
&\geq& C^{\alpha}(\rho_{A_1A_2}^{2\otimes2})[(\frac12)^\frac\alpha2+\frac{(1+k^\omega)^\frac{\alpha}{2}-(\frac12)^\frac{\alpha}{2}}{k^{ \frac{\omega \alpha}{2}}}(\frac{C^2(\rho_{A_1A_3}^{2\otimes2})}{C^2(\rho_{A_1A_2}^{2\otimes2})})^\frac\alpha2]\nonumber\\
&=&(\frac12)^{\frac{\alpha}{2}}C^{\alpha}(\rho_{A_1A_2}^{2\otimes2})+\frac{(1+k^\omega)^{\frac{\alpha}{2}}-(\frac12)^{\frac{\alpha}{2}}}{k^{\frac{\omega\alpha}{2}}}C^{\alpha}(\rho_{A_1A_3}^{2\otimes2}),
\end{eqnarray}
where $0\leq k^\omega\leq k\leq1$, $\omega\geq1$, and $\alpha\geq2$. Combining (32) and (33), one gets (30). Using similar methods, we obtain the inequality (31).
\end{proof}

For a pure $n$-partite quantum state $|\varphi\rangle^{d_1\otimes d_2\otimes \cdots \otimes d_n}\in H_1^{d_1}\otimes H_2^{d_2}\otimes \cdots\otimes H_n^{d_n}$ under bipartite partition $A_1|A_2\cdots A_n$ has the form
\begin{eqnarray}
|\varphi\rangle_{A_1|A_2\cdots A_n}^{d_1\otimes d_2\otimes\cdots\otimes d_n}=\sum_{a_1=1}^{d_1}\sum_{a_2=1}^{d_2}\cdots\sum_{a_n=1}^{d_n}h_{a_1|a_2\cdots a_n}|a_1a_2\cdots a_n\rangle,
\end{eqnarray}
where $h_{a_1|a_2\cdots a_n}\in$ $\mathbb{C}$, $\sum_{a_1a_2\cdots a_n}h_{a_1|a_2\cdots a_n}h_{a_1|a_2\cdots a_n}^\ast=1$. So we get
\begin{eqnarray}
C^2(|\varphi\rangle_{d_1\otimes d_2\otimes\cdots\otimes d_n})=\sum_{a_1,b_1=1}^{d_1}\sum_{a_2,b_2=1}^{d_2}\cdots\sum_{a_n,b_n=1}^{d_n}|h_{a_1|a_2\cdots a_n}h_{b_1|b_2\cdots b_n}-h_{b_1|a_2\cdots a_n}h_{a_1|b_2\cdots b_n}|^2.
\end{eqnarray}

According to (35),
in multipartite quantum systems $\mathcal{H}^{d_1} \otimes\mathcal{H}^{d_2}\otimes\cdots\otimes\mathcal{H}^{d_n}$, the concurrence of the mixed state $\rho_{d_1\otimes d_2\otimes\cdots\otimes d_n}$ satisfies:
$C^2(\rho_{A_1|A_2\cdots A_n}^{d_1\otimes d_2\otimes\cdots\otimes d_n})\geq\sum\frac{1}{(d_1-1)(d_2-1)\cdots(d_n-1)}C^2(\rho_{A_1|A_2\cdots A_n}^{2\otimes2\otimes\cdots\otimes2})$,
where $\sum$ sums over all possible mixed substates $\rho_{A_1|A_2\cdots A_n}^{2\otimes 2\otimes\cdots\otimes 2}$.
For an $n$-qubit quantum states $\rho$ under bipartite partition $A_1|A_2A_3\cdots A_n$, the concurrence satisfies $C^2(\rho_{A_1|A_2A_3\cdots A_n})\geq C^2(\rho_{A_1A_2})+C^2(\rho_{A_1A_3})+\cdots+C^2(\rho_{A_1A_n})$ [4].
Then using the similar method as Theorem 3, we can generalize our result to $n$-partite quantum systems and obtain the following theorem:
\begin{thm}
In multipartite quantum systems $\mathcal{H}^{d_1}_{A_1} \otimes\mathcal{H}^{d_2}_{A_2}\otimes\cdots\otimes\mathcal{H}^{d_n}_{A_n}$, assuming real numbers $0\leq k^\omega\leq k\leq1$, $\omega\geq1$, and $\alpha\geq2$, where $\sum$ represents the sum of all possible mixed substates $\rho_{A_1|A_2\cdots A_n}^{2\otimes 2\otimes\cdots\otimes 2}$ and $\mu=\frac{(1+k^\omega)^{\frac{\alpha}{2}}-(\frac{1}{2})^{\frac{\alpha}{2}}}{k^{\frac{\omega\alpha}{2}}}$, we have the following results:

(1) if $k^\omega C^2(\rho_{A_1A_i}^{2\otimes2})\geq\sum_{j=i+1}^nC^2(\rho_{A_1A_j}^{2\otimes2})$ for $i=2,\cdots,n-1$ and $n\geq3$, then
\begin{eqnarray}
\begin{aligned}
C^\alpha(\rho_{A_1|A_2\cdots A_n}^{d_1\otimes d_2\otimes\cdots\otimes d_n})\geq&\sum[\frac1{(d_1-1)(d_2-1)\cdots(d_n-1)}]^\frac\alpha2[(\frac{1}{2})^{\frac{\alpha}{2}}(C^\alpha(\rho_{A_1A_2}^{2\otimes2})+\mu C^\alpha(\rho_{A_1A_3}^{2\otimes2})\\
&+\cdots+\mu^{n-3}C^\alpha(\rho_{A_1A_{n-1}}^{2\otimes2}))+\mu^{n-2}C^\alpha(\rho_{A_1A_n}^{2\otimes2})].
\end{aligned}
\end{eqnarray}

(2) if $C^2(\rho_{A_1A_i}^{2\otimes2})\leq k^\omega\sum_{j=i+1}^nC^2(\rho_{A_1A_j}^{2\otimes2})$ for for $i=2,\cdots,n-1$ and $n\geq3$, then
\begin{eqnarray}
\begin{aligned}
C^\alpha(\rho_{A_1|A_2\cdots A_n}^{d_1\otimes d_2\otimes\cdots\otimes d_n})\geq& \sum[\frac1{(d_1-1)(d_2-1)\cdots(d_n-1)}]^\frac\alpha2[\mu(C^\alpha(\rho_{A_1A_2}^{2\otimes2})+(\frac{1}{2})^{\frac{\alpha}{2}}C^\alpha(\rho_{A_1A_3}^{2\otimes2})\\
&+\cdots+(\frac{1}{2})^{\frac{(n-3)\alpha}{2}}C^\alpha(\rho_{A_1A_{n-1}}^{2\otimes2}))+(\frac{1}{2})^{\frac{(n-2)\alpha}{2}}C^\alpha(\rho_{A_1A_n}^{2\otimes2})].
\end{aligned}
\end{eqnarray}

(3) if $k^\omega C^2(\rho_{A_1A_i}^{2\otimes2})\geq\sum_{j=i+1}^nC^2(\rho_{A_1A_j}^{2\otimes2})$ for $i=2,\cdots,m$ and $C^2(\rho_{A_1A_i}^{2\otimes2})\leq k^\omega\sum_{j=i+1}^nC^2(\rho_{A_1A_j}^{2\otimes2})$ for $i=m+1,\cdots,n-1$, $\forall$ $2\leq m\leq n-2$, $n\geq4$, we have
\begin{eqnarray}
\begin{aligned}
C^\alpha(\rho_{A_1|A_2\cdots A_n}^{d_1\otimes d_2\otimes\cdots\otimes d_n})\geq&
\sum[\frac1{(d_1-1)(d_2-1)\cdots(d_n-1)}]^\frac\alpha2[(\frac12)^{\frac{\alpha}{2}}(C^\alpha(\rho_{A_1A_2}^{2\otimes2})+\mu C^\alpha(\rho_{A_1A_3}^{2\otimes2})\\
&+\cdots+\mu^{m-2}C^\alpha(\rho_{A_1A_m}^{2\otimes2}))+\mu^m(C^\alpha(\rho_{A_1A_{m+1}}^{2\otimes2})+(\frac12)^{\frac{\alpha}{2}}C^\alpha(\rho_{A_1A_{m+2}}^{2\otimes2})\\
&+\cdots+(\frac12)^{\frac{(n-m-2)\alpha}{2}}C^\alpha(\rho_{A_1A_{n-1}}^{2\otimes2}))+\mu^{m-1}(\frac12)^{\frac{(n-m-1)\alpha}{2}}C^\alpha(\rho_{A_1A_n}^{2\otimes2})].
\end{aligned}
\end{eqnarray}
\end{thm}

{\bf Remark 3.}
We have presented the monogamy inequalities of the $\alpha$th ($\alpha\geq2$) power of concurrence. By adopting the inequalities (6) of Lemma 1 and the above method, we can also obtain similar results for $0\leq\alpha\leq2$ which cover all the real numbers. As Remark 2 shows, we have found tighter results for $0\leq\alpha\leq2$.

\section{Conclusion}
In this article, we have obtained the monogamy inequalities of the Bures measure of entanglement and the geometric measure of entanglement for tripartite quantum states and $n$-qubit quantum states.
Using examples we have shown that our monogamy relations are tighter than the existing ones.
Moreover, we have discussed the monogamy inequalities of concurrence in higher-dimensional quantum systems which give rise to finer characterizations of the entanglement shareability and distribution among qudit systems.
Our results may help us understand better the monogamy nature of multipartite higher-dimensional quantum entanglement.

\textbf {Acknowledgements}
This work is supported in part by Simons Foundation under grant no. 523868 and NSFC under grant nos. 12126351 and 12126314.

\end{CJK*}
\end{document}